\documentclass[10pt,journal,compsoc]{IEEEtran}

\ifCLASSOPTIONcompsoc
  \usepackage[nocompress]{cite}
\else
  \usepackage{cite}
\fi

\ifCLASSINFOpdf
\else
\fi

\usepackage{times}  
\usepackage{helvet}  
\usepackage{courier}  
\usepackage{url}  
\usepackage{graphicx}  
\usepackage{multicol}

\usepackage{times}
\usepackage{xcolor}
\usepackage{soul}
\usepackage[utf8]{inputenc}
\usepackage[small]{caption}

\usepackage{wrapfig}
\usepackage{graphicx}
\usepackage{amsmath}
\usepackage{amssymb}
\usepackage{amsthm}
\usepackage{multirow}
\usepackage{lipsum}
\usepackage{colortbl}
\usepackage{cutwin}

\usepackage{algorithmic}
\usepackage{color}
\usepackage{verbatim}

\usepackage{cleveref}

\newtheorem{proposition}{Proposition}

\newtheorem{definition}{Definition}

\DeclareMathOperator*{\argmax}{arg\,max}

\newcommand\T{\rule{0pt}{2.1ex}}
\newcommand\B{\rule[-0.7ex]{0pt}{0pt}}

\definecolor{Gray}{gray}{0.8}

\title{\vspace{2mm}{Effectiveness of Diffusing Information through\\a Social Network in Multiple Phases}}

\author{
Swapnil Dhamal
 \IEEEcompsocitemizethanks{
        \vspace{1mm}
         \IEEEcompsocthanksitem Swapnil  Dhamal is a postdoctoral researcher with Institut National de Recherche en Informatique et en Automatique (INRIA), Sophia Antipolis-M\'editerran\'ee, France.
         }
}

\begin{document}

%

\IEEEtitleabstractindextext{
\vspace{2mm}
\begin{abstract}
We study the effectiveness of using multiple phases for maximizing the extent of information diffusion through a social network, and present insights while considering various aspects. In particular, we focus on the independent cascade model with the possibility of adaptively selecting seed nodes in multiple phases based on the observed diffusion in preceding phases, and conduct a detailed simulation study on real-world network datasets and various values of seeding budgets. We first present a negative result that more phases do not guarantee a better spread, however the adaptability advantage of more phases generally leads to a better spread in practice, as observed on real-world datasets. We study how diffusing in multiple phases affects the mean and standard deviation of the distribution representing the extent of diffusion. We then study how the number of phases impacts the effectiveness of multiphase diffusion, how the diffusion progresses phase-by-phase, and what is an optimal way to split the total seeding budget across phases. Our experiments suggest a significant gain when we move from single phase to two phases, and an appreciable gain when we further move to three phases, but the marginal gain thereafter is usually not very significant. Our main conclusion is that, given the number of phases, an optimal way to split the budget across phases is such that the number of nodes influenced in each phase is almost the same.
\end{abstract}
\vspace{2mm}
}

\maketitle

\IEEEdisplaynontitleabstractindextext

\IEEEpeerreviewmaketitle

\IEEEraisesectionheading{\section{Introduction}}
\label{sec:mp_intro}

\IEEEPARstart{O}{wing} to the advent of online social networks, diffusing information to a large population in a short span of time has become a reality.
 Product companies (or campaigners) use this fact to their advantage by harnessing social networks for viral marketing wherein, they offer free or discounted product samples (or present some information) to a selected set of individuals, who advertize the product (or forward the information) to their friends. If these friends buy the product and like it, they likely advertize it to their friends, and this process goes on.
 A fundamental problem with respect to this idea is to select the nodes to whom free or discounted product samples are to be given (often referred to as {\em seed nodes\/}), such that the number of individuals influenced by the product marketing (often referred to as {\em extent of diffusion\/}) is maximized.
 This leads to an important optimization problem: given a budget on the number of seed nodes, which nodes should be selected for seeding so that the extent of diffusion is maximized?
 This problem is often referred to as the problem of {\em influence maximization in a social network\/}.

 Owing to the inherent uncertainties of social networks due to  uncertainties  in human behavior, transmission of information, interaction frequencies, etc., the process of information diffusion is an uncertain one. It is thus a stochastic process, with several possibilities of instances with certain probabilities.
  A widely used metric of performance of an influence maximizing seed-selection algorithm is the expected extent of diffusion taken over all instances (or a large number of instances if computing expectation over all instances is computationally hard).  
 However, selecting a set of seed nodes may lead to an excellent extent of diffusion in one instance, while a very poor extent in another. 
 This has motivated researchers to consider the possibility of adaptive seeding, which could reduce the uncertainty involved by adaptively selecting seed nodes based on the diffusion observed thus far.
In order to understand the advantages and disadvantages of multiphase diffusion, we first present some preliminaries.

\subsection{Preliminaries}

Consider a social network $G$, with $N$ as its set of $n$ nodes and $E$ as its set of $m$ weighted directed edges.

\subsubsection{Independent Cascade (IC) model}


Each directed edge $(u,v) \in E$ has an associated weight indicating the influence probability $p_{uv}$ (the probability with which node $u$ would influence node $v$, if node $u$ is influenced). 
The diffusion progresses in discrete time steps. At time  0, the selected seed nodes are influenced. At time step 1, each seed node $u$ independently attempts to influence each of its neighbors $v$ and succeeds with probability $p_{uv}$. In time step 2, all the nodes that were influenced in time step 1 independently attempt to influence their respective neighbors and succeed with the corresponding influence probabilities. 
This process continues until no further nodes can be influenced.
The expected extent of diffusion can  be obtained by taking a weighted average of the number of nodes influenced over all possible diffusions using IC model (where  weight is the probability of progressing according to the corresponding diffusion).

\subsubsection{Live Graph}

A live graph $L$ is an instance of  $G$, obtained by sampling  edges with  corresponding edge influence probabilities.
A live graph, being an instance, is a directed but unweighted graph.
An edge $(u,v)$ is present in it with probability $p_{uv}$ and absent with probability $1-p_{uv}$, independent of the presence of other edges.
The probability of occurrence of a live graph $L$ is thus,
$\prod_{(u,v) \in L} (p_{uv}) \prod_{(u,v) \notin L} (1-p_{uv})$.
Kempe, Kleinberg, and Tardos
\cite{kempe2003maximizing}  show that,
since   influence probabilities do not change with time, sampling an edge $(u,v)$ in the beginning of  diffusion is equivalent to sampling it when $u$ is influenced. 
The expected extent of diffusion starting from a set of seed nodes, can thus be defined as a weighted average of the number of nodes reachable from that set over all  live graphs (where  weight is the probability of occurrence of the corresponding live graph).

\subsubsection{Multiphase Information Diffusion using IC Model}

Let $p$ be the number of phases for which the diffusion is planned to run.
After the selection of a certain number of seed nodes in the first phase using an influence maximizing algorithm, the diffusion starts and progresses according to IC model, until no further nodes can be influenced. 
Then based on the observed diffusion thus far, the network could be modified by removing nodes which are already influenced, since they would play no further role in changing the diffusion state of the network. 
This modified network can be viewed the {\em diffusion state} of the original network after the first phase.
Now  a certain number of seed nodes are selected for the second phase using the influence maximizing algorithm on this modified network, following which, the diffusion progresses until no further nodes can be influenced. This process repeats until the termination of the last phase (phase $p$).
%
%
%
Note that we could have initiated a phase before the termination of the preceding phase, however it would partially nullify the purpose of using multiple phases by not observing the diffusion till its completion. As our primary goal is to study  effectiveness of multiphase diffusion, we consider  usage of multiple phases at their full potential by allowing the diffusion in a phase to terminate before initiating the next phase.

We now provide an intuition why multiphase diffusion would be  advantageous as compared to single diffusion.
An influence maximizing algorithm which is intended to maximize the expected number of influenced nodes, would possibly not select an influential node if it is likely to get influenced owing to other already selected seed nodes with high probability. But unless this high probability is equal to 1, there will exist `bad' live graphs in which the influential node would not get influenced. Adaptive seeding would select this node as a seed node if our observed diffusion indicates that the underlying live graph is `bad'. This would thus improve the extent of diffusion in expectation.
On the other hand, the algorithm may have selected a node because it is influential enough, but not likely to get influenced owing to other already selected seed nodes. Again, there would exist live graphs in which this node gets influenced without having to select it as seed node. In such live graphs, adaptive seeding would instead select another node which did not actually get influenced in our observed diffusion, which again would lead to a higher expected extent of diffusion.

A  drawback of multiphase diffusion is that the diffusion may progress at a  slower rate owing to the delay in selecting seed nodes in subsequent phases.
Like in most of the literature, we consider that this delay  does not impact the value of our  diffusion; we provide a note on accounting for this delay at the  end of the paper.






Given a total budget of $k$ that is to be distributed across $p$ phases,
let $k_q$ be the budget allotted for phase $q$.

\begin{definition}[Budget split]
A budget split is a vector representing how the total budget is allotted for different phases.
\end{definition}

So for an information diffusion process that is executed over $p$ phases, the budget split can be represented as $(k_1,\ldots,k_p) = (k_q)_{q=1}^p$.
We use $\mathbf{K}$ to denote a budget split.

Since the total budget is $k$, we should have that 
$
\sum_{q=1}^p k_q \leq k
$.
Note that if there is any surplus budget $(k-\sum_{q=1}^p k_q)$, this surplus can be used up in the terminal phase to influence additional nodes which could not be influenced at the end of phase $p$. So it is optimal to have the above constraint tight, that is,
$
\sum_{q=1}^p k_q = k
$.

\begin{definition}[Optimal budget split]
Given an influence maximizing algorithm and a budget for a network, an optimal budget split is one that maximizes the expected extent of diffusion achieved over all phases combined.
\end{definition}

Given an influence maximizing algorithm and budget $k$ for a  network,
let $\beta_q(\mathbf{K})$ be the expected extent of diffusion or the expected number of nodes influenced in phase $q$, if the budget split is $\mathbf{K}$.
So the expected extent of diffusion over $p$ phases is $\sum_{q=1}^p \beta_q(\mathbf{K})$. An optimal budget split is, thus,

\begin{small}
\begin{align*}
\mathbf{K}^* = (k_1^*,\ldots,k_p^*) = \argmax_{\mathbf{K}} \sum_{q=1}^p \beta_q(\mathbf{K})
\end{align*}
\end{small}

\subsection{Relevant Work}
\label{sec:mp_relevant}

The problem of maximizing information diffusion in social networks was first studied from  algorithmic and computational viewpoint by
Kempe, Kleinberg, and Tardos
\cite{kempe2003maximizing}, where they showed $\left( 1-\frac{1}{e} \right)$-approximation guarantee of greedy algorithm for selecting seed nodes.
However, it is computationally infeasible to run this algorithm on large social networks.
Several alternatives have been proposed to bypass this computational barrier. 
Goyal, Lu, and Lakshmanan
\cite{goyal2011celf} present
a lazy forwarding approach to avoid unnecessary computations made in the greedy algorithm.
Chen, Wang, and Yang
\cite{chen2009efficient}
present 
a number of efficient versions of the greedy algorithm and also propose a very fast degree discount heuristic.
Wang, Chen, and Wang
\cite{wang2012scalable}
propose a fast heuristic (PMIA) based on the concept of maximum influence arborescence, and show that it performs very close to the greedy algorithm on real-world social network datasets.
Jung, Heo, and Chen
\cite{jung2012irie} propose an even faster high-performance heuristic (IRIE) by integrating influence ranking and influence estimation, 
making it feasible to run on networks with tens of millions of edges.

Adaptive seeding or diffusing information through a social network in more than one phase is a relatively nascent area.
Golovin and Krause
\cite{golovin2010adaptive,golovin2011adaptive}
introduce the concept of adaptive submodularity
and prove that any problem satisfying this property facilitates an adaptive greedy algorithm to provide an approximation guarantee;
they show that this property is satisfied for adaptive seeding in IC model.
For
 adaptive seeding  with any monotone submodular
function,
Badanidiyuru et al.
\cite{badanidiyuru2016locally} 
propose an approximation algorithm  based on locally-adaptive policies.
%
Specific to influence maximization in social networks,
Singer
\cite{singer2016influence}
presents a survey on  adaptive seeding methodologies.

Seeman and Singer
\cite{seeman2013adaptive} were among the first to dedicatedly study the adaptive seeding framework.
Rubinstein, Seeman, and Singer
\cite{rubinstein2015approximability}
study the approximability of adaptive seeding algorithms that incentivize
nodes with heterogeneous activation costs.
Horel and Singer 
\cite{horel2015scalable} develop scalable methods for adaptive selection of the seed set with provable guarantees for models in which the influence of a set can be expressed as the sum of the influence of its members. However, these methods do not apply to  IC-like models.
Correa et al. 
\cite{correa2015adaptive}
show that in the homogeneous case (where every pair of nodes randomly meet at the same rate), the adaptivity benefit   is bounded by a constant. 

Dhamal, Prabuchandran, and Narahari
\cite{dhamal2phase} show a trade-off between the size of the observed diffusion
and the exploitation based on the observed 
diffusion, 
while splitting the budget between two phases.
%
%
%
Tong et al.
\cite{tong2016adaptive}
 study adaptive seeding in the dynamic IC model, and provide  performance guarantee of the greedy adaptive seeding algorithm.
%
%
%
%
Yuan and Tang
\cite{tang2016no}
present a theoretical study of a framework where seed node can be selected before the ongoing diffusion terminates, and hence develop a policy that achieves a bounded approximation ratio.
%
%
Mondal, Dhamal, and Narahari
\cite{mondal2017two}
study the influence maximization problem in two phases, where the first phase is regular diffusion and the second phase is boosted  using referral incentives.


\subsection{Contributions}

To the best of our knowledge, this is the first work to study information diffusion in more than two phases, and present insights on the distribution, phasewise progression, and optimal budget split.
Our specific contributions are as follows:

\begin{itemize}
\item 
We present a negative result that more phases do not guarantee a higher extent of diffusion.
\item
Using real-world network datasets, we study how diffusing in multiple phases affects the mean and standard deviation of the distribution representing extent of diffusion.
\item
We study the effectiveness of multiphase diffusion with respect to the number of phases, and the phase-by-phase progression of diffusion so as to quantify the delay in diffusion owing to the delayed selection of seed nodes.
\item 
We develop a method for determining an optimal budget split for a given number of phases, based on the nature of the underlying network.
\end{itemize}



\section{Problems in Multiphase Diffusion}



Consistent with almost all studies on adaptive seeding, we assume that seed nodes are selected in a given phase without considering their eventual impact on the next phase, that is, seed nodes are selected in phase $q$ by using a single phase optimal policy with the reduced budget of $k_q$, without accounting for the presence of phase $q+1$. 
This approach is termed as {\em myopic approach\/} by Dhamal, Prabuchandran, and Narahari \cite{dhamal2phase}.
In their study, the authors show that using farsighted approach (selecting nodes in a phase by considering its impact on the next phase)
with \textbf{any} budget split
would always lead to a better extent of diffusion in expectation, but do not support or oppose any statement regarding whether the myopic multiphase approach would always outperform single phase.
We fill this gap by showing a negative result.

\subsection{A Negative Result}

Firstly, it can be easily seen that more phases may not be advantageous if the budget split is not made judiciously. 
For instance, a 2-phase budget split of $(\frac{1}{3}k,\frac{2}{3}k)$ would most certainly be better than a 3-phase budget split of $(k-2,1,1)$ for reasonably large  $k$, since the latter would perform  close to single phase, while a $(\frac{1}{3}k,\frac{2}{3}k)$ split would have a significant gain over single phase \cite{dhamal2phase}.
However, one could ask: would having more phases by subdividing allocation of an existing phase, result in a better performance?
%
%
We  show the answer is negative using a simple counterexample.
According to \cite{dhamal2phase},
using a farsighted approach, a budget split of $(1,1)$ would perform at least as good as single phase with $k=2$ on any network.
We show that this is not guaranteed using a myopic approach, that is, there exists a network for which we could come up with a 2-phase budget split that performs worse than single phase.

\begin{proposition}
Replacing budget split $(\ldots,k_q,\ldots)$ with $(\ldots,x,k_q-x,\ldots),x \in \{1,\ldots,k_q-1\}$ may lead to a worse extent of diffusion, even with respect to an optimal policy.
\label{prop:negative}
\end{proposition}

\begin{proof}
~
\\
\begin{minipage}{.32\textwidth}
We show this with a  counterexample for $p=2,q=1$ and $k=k_q=2,x=1$.
\\
The edge influence probabilities are: $p_{AC}=p_{BC}=0.5 \,,\, p_{CD}=p_{CE}=1$.
The example graph is alongside.
\end{minipage}
\begin{minipage}{.14\textwidth}
\centering
\includegraphics[scale=.35]{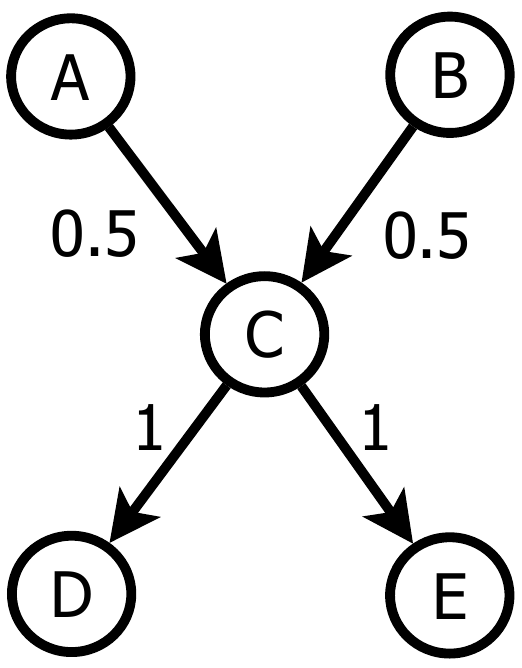}
\end{minipage}
 
  \vspace{3mm}
When using single phase with $k=2$, the optimal solution to maximize expected diffusion is to select $A$ and $B$ as the two seed nodes.
Node $C$ would then get influenced with probability $1-(1-0.5)(1-0.5)=0.75$.
Since nodes $D$ and $E$ would get influenced with probability $1$ if node $C$ is influenced, we have that these two nodes would get influenced with probability $0.75$ each.
So with $A$ and $B$ as the two seed nodes, the expected number of nodes influenced at the end of the diffusion process is $1+1+0.75+0.75+0.75 = 4.25$.
It can be easily seen that selecting any other set of two nodes would lead to a lower expected extent of diffusion, for instance, $\{C,(A \text{ or } B)\}$ leads to $4$, $\{C,(D \text{ or } E)\}$ leads to $3$, $\{(A \text{ or } B),(D \text{ or } E)\}$ leads to $3$, and $\{D,E\}$ leads to $2$.
So using optimal policy, the expected number of influenced nodes at the end of single phase diffusion process is $\mathbf{4.25}$.

When using two phases with budget split $(1,1)$, that is, $k_1=k_2=1$, the optimal node to select in the first phase is $C$, which leads to nodes $C,D,E$ getting influenced with probability $1$ in  first phase. So the expected number of nodes influenced at the end of  first phase is $3$. Selecting any other node would lead to a lower number of influenced nodes, for instance, $A$ or $B$ would lead to $2.5$, and $D$ or $E$ would lead to $1$.
Since node $C$ is the optimal choice for  first phase, we know with certainty that at the start of  second phase, nodes $C,D,E$ are already influenced. Since $k_2=1$, it is optimal to select either node $A$ or $B$ as the seed node for  second phase, which would lead to exactly $1$ additional node getting influenced in  second phase. 
So at the end of this two-phase diffusion process, the number of influenced nodes is $\mathbf{4}$, which is lower than that achieved using single phase ($4.25$).
\end{proof}



%



The above negative result is of theoretical interest, however, it has been shown that adaptive seeding performs better than single phase seeding on real-world  network datasets.
It is also known that an adaptive method of diffusing information under IC model preserves the approximation guarantee of $\left( 1-\frac{1}{e} \right)$ provided by  greedy  algorithm \cite{golovin2010adaptive,golovin2011adaptive}.
In this paper, we use the state-of-the-art IRIE algorithm, which performs very close to greedy algorithm while running several order of magnitude faster.

\subsection{Problems of Interest}

\subsubsection{Distribution of the Extent of Diffusion}


Information diffusion under IC model being a stochastic process, there are uncertainties involved regarding how the diffusion progresses. This is an important selling point of multiphase diffusion, since it reduces the uncertainty while selecting seed nodes in subsequent phases.
It is a well-accepted practice in the literature to derive conclusions regarding the performance of an algorithm by only considering the expected number of nodes influenced at the end of the diffusion process (or the expected extent of diffusion).
However,
it would be interesting to study how using multiple phases affects the entire distribution of the extent of diffusion, instead of its expected value alone. In particular, we could plot the distributions, observe their nature, and also study the implications of their standard deviations (in addition to their means).
%
%
Further, it would  be interesting to see how the distribution changes from phase to phase, and what it actually means when we say that multiphase diffusion  reduces uncertainty. 





\subsubsection{Impact of the Number of Phases}


It has been a consistent result in the two-phase diffusion and adaptive seeding literature that using two phases yields a significant gain over single phase. A natural question arises regarding how beneficial going beyond two phases would be.
A primary objective of using social network for diffusing information, is to enable the information to reach as many individuals as possible. 
But it may also be important that the information reaches the individuals as early as possible, especially in the presence of a competing information or when the value of  information decreases with time.
So if using $p+1$ phases instead of $p$ phases (with the same total budget $k$) improves the extent of diffusion negligibly, it may be well advised to not increase the number of phases.
Further, additional phase may incur additional costs.
This motivates us to study how the amount of gain changes as we increase the number of phases.

\subsubsection{Determining Optimal Budget Split}

Multiphase diffusion relies on the fact that we observe diffusion at the end of a phase, and exploit this observation by adaptively selecting seed nodes in the following phase.
An optimal way to split the total budget is thus important 
to find an optimal balance between observation and exploitation.
Hence the effectiveness of a $p$-phase diffusion fundamentally depends on how the total budget is split among the $p$ phases.
%
In order to generalize our findings to general social network datasets, it is also important to identify patterns and insights behind the observed optimal budget splits.


\subsubsection{Progression of Diffusion with Phases}

As mentioned earlier, though
our primary objective is to reach or influence as many individuals as possible, it is also important that the information reaches them as early as possible, especially in the presence of a competing information or when the value of the information or product decreases with time.
So given that we are diffusing information across a total of $p$ phases, it is important to know how many nodes get influenced at the end of each phase.
Several marketing, pricing, or campaigning decisions may be impacted with the knowledge of how diffusion progresses over its different phases.
For instance, if the product value decreases over time, a company may be willing to compromise on the optimal budget split and hence the final extent of diffusion, so as to have a higher extent of diffusion during the early phases.

\section{Simulation Setup}


\subsection{Simulation Technique}

%

We first discuss a naive approach of simulating multiphase diffusion, explain its drawbacks,  then present our approach.

\subsubsection{A Naive Approach}
Starting with $k_1$ best seed nodes,
the simulations are first run for $\mathcal{M}_1$ Monte Carlo iterations, each according to IC model, to arrive at $\mathcal{M}_1$ possible diffusion states at the end of phase 1.
For each of these $\mathcal{M}_1$ states, we then adaptively select $k_2$ best seed nodes and run the simulations for $\mathcal{M}_2$ iterations to arrive at $\mathcal{M}_2$ diffusion states.
So now we have a total of $\mathcal{M}_1 \mathcal{M}_2$ diffusion states at the end of phase 2.
Continuing thus, we have $\prod_{q=1}^p \mathcal{M}_q$ diffusion states at the end of phase $p$.
If we run the simulations for $10^4$ iterations (as run for single phase in most papers in the literature) in each phase, 
we need to run the influence maximizing algorithm on $1$ state (the given graph) in the first phase, $10^4$ states in second phases, $10^{8}$
 states in third phases, and so on.
In addition to selecting seed nodes, simulating diffusion using IC model also would add considerably to the running time;
we need to run the diffusion process $10^4$ times in the first phase, $10^8$ times in second phase, $10^{12}$
 times in third phases, and so on.
So it is clear that the branching process grows exponentially and it is rather infeasible to run the simulations on large datasets over large number of phases.





\subsubsection{Our Approach}


We presample a set of $\mathcal{M}$ live graphs before the diffusion starts, instead of determining the presence of each edge $(u,v)$ in live graph after $u$ is influenced. We then use these as a common set of live graphs across various simulations.
This idea similar to \cite{chen2009efficient} wherein live graphs are predetermined to enable precomputations of reachability from any given node and hence avoid repeated computations during program execution.
Such an approach can be justified by considering that the underlying live graph already exists (but known to us only probabilistically) and is uncovered during  diffusion process.
Also, finding reachability from a set of nodes in live graph is equivalent to diffusing information starting from these nodes \cite{kempe2003maximizing}.

By presampling live graphs, the reachability from every node in every live graph is computed once and stored.
So its highlight is that we do not have to simulate diffusion using IC model each time; only retrieve the stored set of reachable nodes
\cite{chen2009efficient}.
Another advantage of presampling a common set of live graphs for all simulations (for different budget splits and also different number of phases) is that, we can not only compare their performances, but also reliably draw conclusions regarding aspects such as means and standard deviations of extents of diffusion during and after each phase, by comparing their distributions.







In our simulations, we set
$\mathcal{M}=10^4$.
For the datasets considered (enlisted later), this count of Monte Carlo simulations gave precise results (that is, running independent sets of $10^4$ Monte Carlo simulations lead to extents of diffusion with almost equal means and standard deviations).


\subsection{Extending Algorithm to Multiple Phases}

We use IRIE as our influence maximizing algorithm for determining seed nodes.
To the best of our knowledge, this is the best known algorithm for its performance (very close to the greedy algorithm) as well as running time (few seconds for a graph with million edges).
In our simulations, we set the damping factor $\alpha=0.7$ as identified by the authors to be the value for which IRIE's accuracy is found to be the highest.
In the first phase, we run IRIE just as for single phase, albeit with a budget of $k_1$. The reachability from the selected $k_1$ seed nodes in each of the $\mathcal{M}$ live graphs lead to the corresponding $\mathcal{M}$ diffusion states.
Subsequently, for $q$ ranging from $1$ to $p$, we select $k_q$ seed nodes in phase $q$ for each of the diffusion states, which then after considering reachability from these newly selected nodes in the corresponding live graphs, lead to $\mathcal{M}$ new/updated diffusion states (which act as the starting point for phase $q+1$).
%

Hence the number of diffusion states for which we run IRIE is 
\mbox{$(p-1)\mathcal{M}+1$} (including the starting state, the given graph itself).
It is to be noted that we run IRIE independently on these states,
that is, we use IRIE as a black box. We do not eliminate the possibility of adapting IRIE in a better way for diffusion in multiple phases; we defer this to future work.
The running time of IRIE approximately increases with the number of edges in the network. So the overall running time of the entire multiphase seed selection algorithm for a given budget split is proportional to
\mbox{$|E|(p-1)\mathcal{M}$}.

\begin{figure*} 
\small
\begin{tabular}{ccc}
\hspace{-8mm}
\includegraphics[scale=.44]{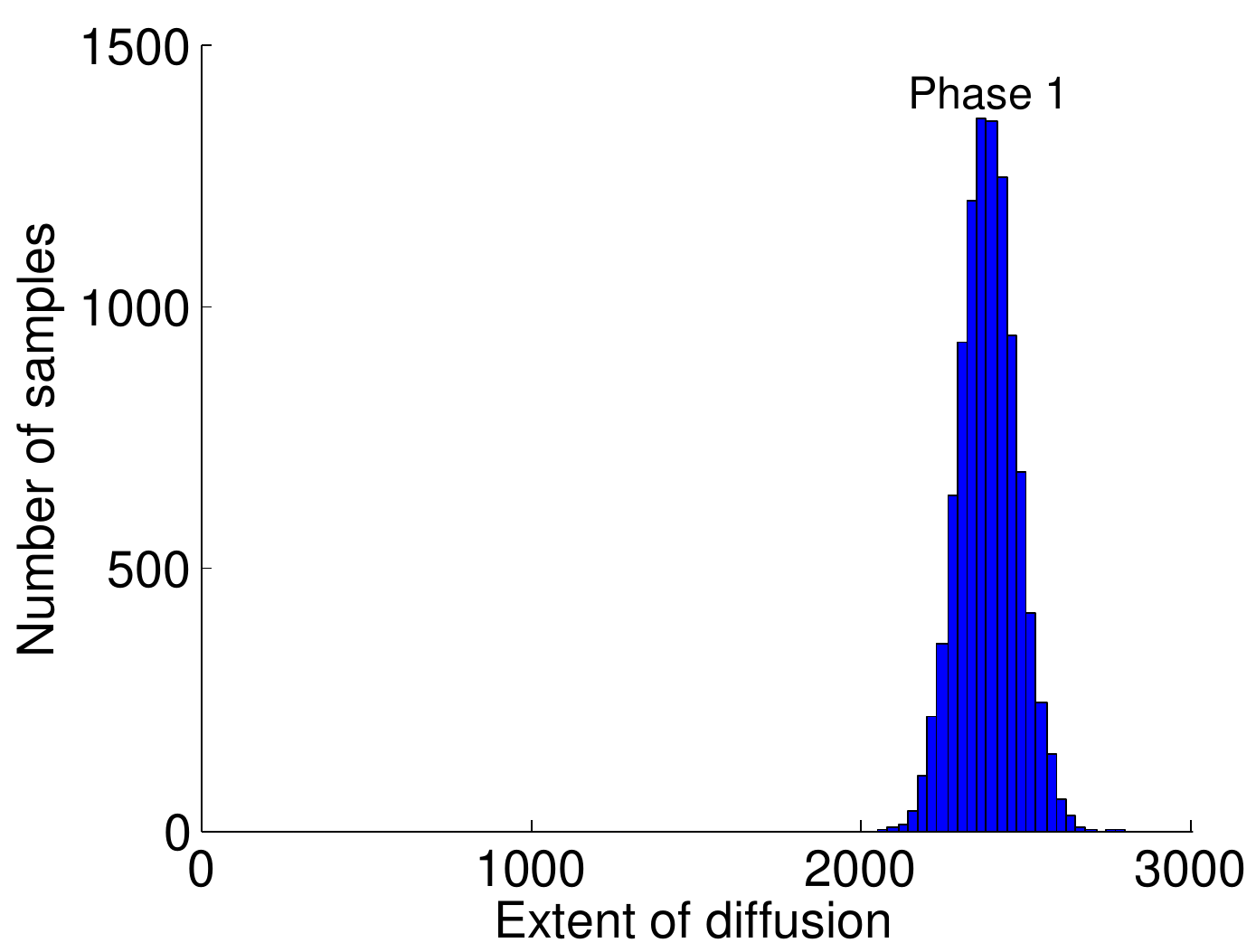}
&
\hspace{-7mm}
\includegraphics[scale=.44]{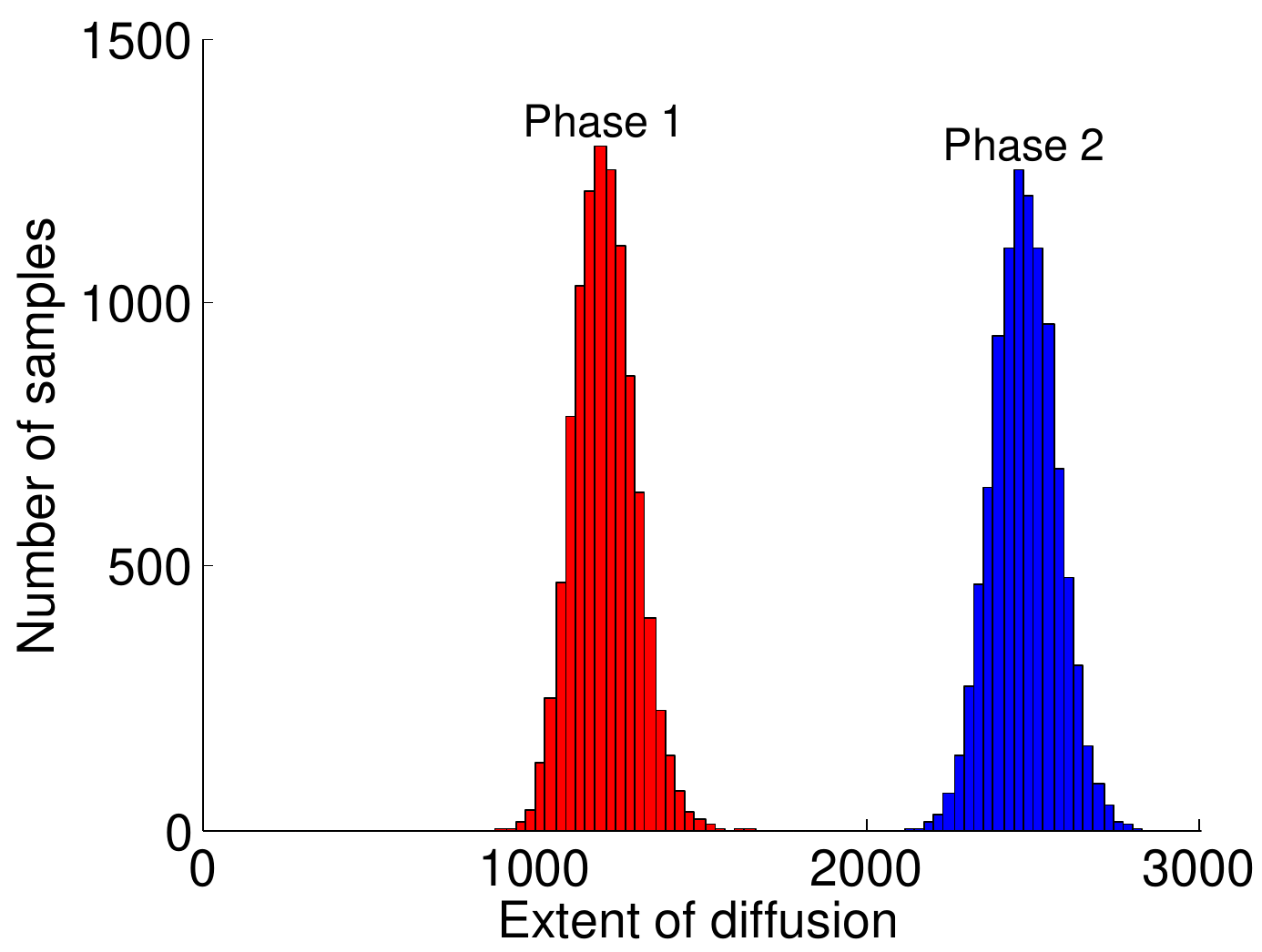}
&
\hspace{-6mm}
\includegraphics[scale=.44]{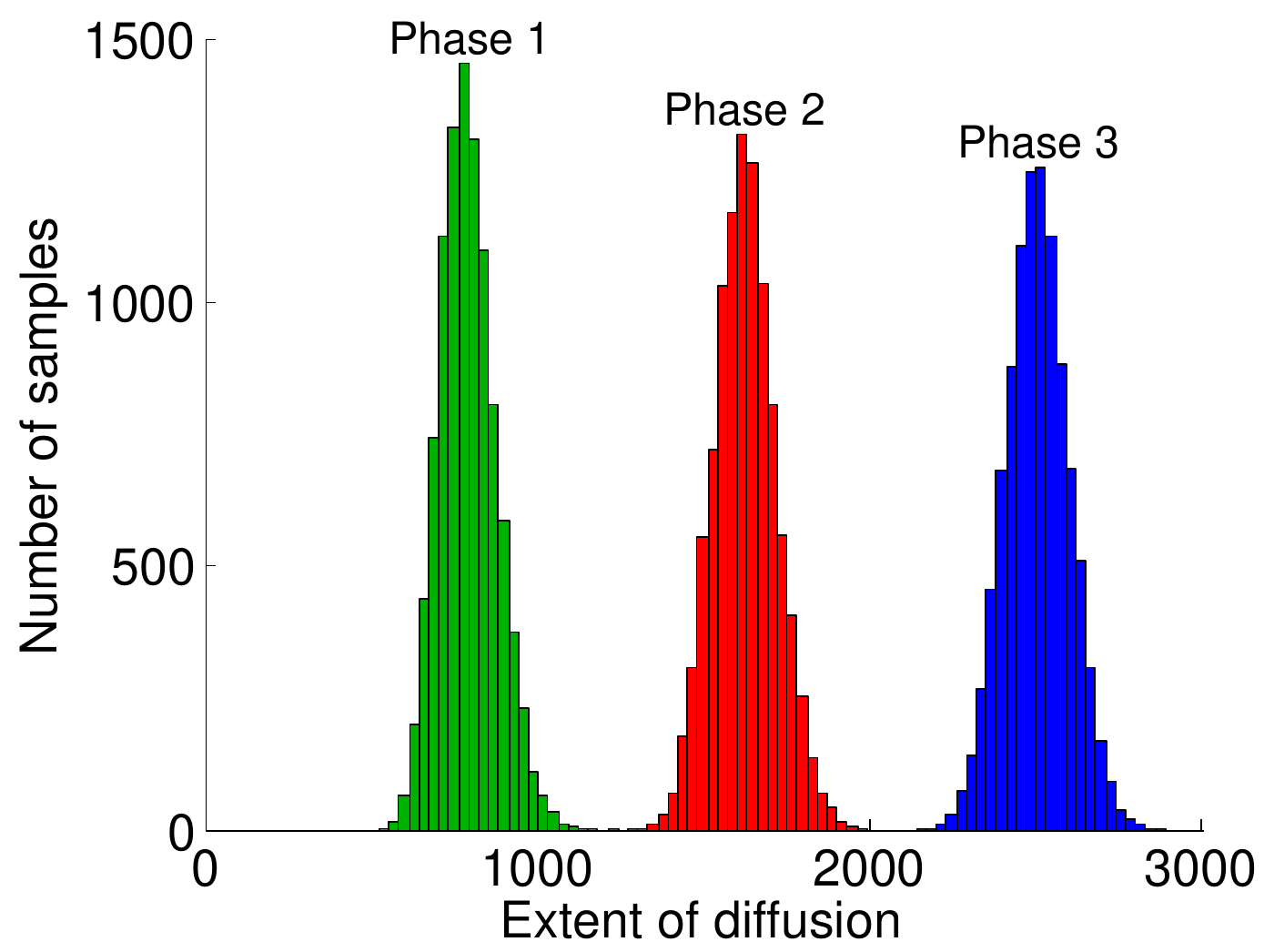}
\\
(a) Single phase diffusion
&
(b) Two-phase diffusion
&
(c) Three-phase diffusion
\end{tabular}
\caption{Distribution of extent of diffusion for different number of phases for NetHEPT (WC) with $k=200$
}
\label{fig:dist}
\end{figure*}

\subsection{Searching for Optimal Budget Split}

Given a budget $k$, the budget $k_q$ allotted for phase $q$ can take $k+1$ possible values (including $0$).
Since there are $(p-1)$ degrees of freedom owing to  constraint \mbox{$\sum_{q=1}^p k_q = k$},
the number of points (corresponding to possible budget splits) in the
standard discrete simplex is 
{\tiny{$\dbinom{k+p-1}{p-1}$}}.
So it is  infeasible to exhaustively search over all  budget splits even for relatively small values of $p$, for practical values of $k$. 


It is a general observation in the literature that 
the extent of diffusion usually turns out to be a smooth function of the budget, that is, a slight change in budget usually does not result in a drastic change in the extent of diffusion.
We harness this to avoid exhaustively searching over all budget splits.
We search for an optimal budget split in two steps, by first doing a coarse search (looking at a small number of well-separated budget splits) and then a fine search (looking in the neighborhood of good-valued budget splits found in our coarse search).
We later briefly discuss how one could improve on this search technique for our particular problem.

In our coarse search, we  assign each $k_q$ a value from  $\{0, 0.1k, 0.2k, \ldots, k\}$ (11 values) such that 
\mbox{$\sum_{q=1}^p k_q = k$}. 
%
The number of points (possible budget splits) in the 
 standard discrete simplex is now
{\tiny{$\dbinom{(11-1)+(p-1)}{p-1}=\dbinom{p+9}{p-1}$}}. For $p=5$, this equals $1001$, which is still a  large search space.
However, we could reduce it by noting that, if $k_q=0$ for some $q$, it is equivalent to having less than $p$ phases. Also, several budget splits would be equivalent, for instance, the 3-phase budget splits $(k_x,k_y,0),(k_x,0,k_y),(0,k_x,k_y)$ are all equivalent to the 2-phase budget split $(k_x,k_y)$.
So the results for such budget splits where $k_q=0$ for some $q$, can be directly appended from the results obtained for less than $p$ phases.
So in our coarse search, we only look at budget splits where $\forall q, k_q>0$ and is an integral multiple of $0.1k$.
This is equivalent to slicing a bar of length $k$ into $p$ pieces by making $p-1$ cuts at integral multiples of $0.1k$.
Since there are $9$ possible locations where we can make these $p-1$ cuts, the number of  ways in which we can make these cuts is
{\tiny{$\dbinom{9}{p-1}$}}. 
This is a valid value
since our simulations have $p\leq 5$.
For $p=5$, this equals $126$, which is a tractable search space.

Following the coarse search, we do a fine search for budget allocations in multiples of $0.05k$ (rounded below, if required).
On finding the budget split vector giving  maximum extent of diffusion (probability of finding multiple maxima is 0), say $(h_q)_{q=1}^p$, we look for the budget splits obtained by incrementing and decrementing its individual coordinates by $0.05k$.
Note that since we have $p-1$ degrees of freedom, we are looking at a $(p-1)$-dimensional space, and so the number of increments and decrements for all free coordinates combined is $2(p-1)$. We also check throughout that  budget allocation stays non-negative for the constrained coordinate.

Now for each of the $p-1$ dimensions, we have values obtained by incrementing, decrementing, and not changing the coordinate, one of which gives the maximum value among the three; let this coordinate be $z_q$. We now can form a hypercube whose vertices have the $q^{th}$ coordinate as $h_q$ or $z_q$. Note that this could be a less than $(p-1)$ dimensional hypercube if $h_q=z_q$ for some $q$'s. The vertices of this hypercube are now the new budget splits we search on. 
If a budget split is already explored, we recall its stored value.
So the maximum number of new budget splits derived using this hypercube (which is when $h_q \neq z_q$ for all free coordinates) would be $2^{p-1}-p$.
This concludes our fine search in the neighborhood of the best budget split that was obtained using coarse search. 
The maximum number of new budget splits found is thus, $2^{p-1}-p+2(p-1)$; this equals $19$ for $p=5$. We compute  expected extent of diffusion for each of them.
We follow this method for the second best upto the  tenth best budget split (run in parallel on different machines independently; so computations  repeated for some budget splits).



We employ the above search method for $p \geq 3$; for $p=2$, we  search through all multiples of $0.05k$.
In addition to the budget splits searched as above, we explore budget splits of certain manually chosen ratios, which we see later.


\subsection{Datasets Used}

        The study of multiphase diffusion is computationally very intensive in nature,
owing to the large number of intermediate diffusion states after each phase on which we need to run the seed selection algorithm, as well as owing to
the large number of possible budget splits we need to take into consideration.
So with the computational power available to us, it was infeasible to run the multiphase simulations on very large datasets studied in the literature for single phase diffusion.
So we focus our simulation study on moderate sized datasets (which are commonly used in the literature) to draw conclusions and provide insights
based on our observations.

We conduct extensive simulations for upto 5 phases on NetHEPT dataset [$|V|=15K, |E|=31K$].
%
This dataset has been extensively used for experimentation in the literature \cite{kempe2003maximizing,chen2009efficient,wang2012scalable}. 
We also conduct simulations on Facebook dataset [$|V|=4K, |E|=88K$] \cite{leskovec2012learning} for upto 4 phases.

For modeling edge influence probabilities in networks,
we use two widely accepted ways, namely, the {\em weighted cascade (WC) model} and the {\em trivalency (TV) model} \cite{wang2012scalable,jung2012irie}.
In WC model, for every edge $(u,v)$ in the network, $p_{uv}$ is equal to the reciprocal of $v$'s degree.
In TV model, every edge in the network is assigned a probability value that is uniformly sampled from the set of values $\{0.001, 0.01, 0.1\}$.

In addition to studying NetHEPT with a budget of $k=50$ (like in most of the literature), we also look at  $k=200$ (like in \cite{dhamal2phase}) since it would allow each individual phase to have enough budget to show an impact when the number of phases is large.
Also, studying different values of budgets would allow us to identify any patterns and draw more reliable conclusions.

\section{Simulation Results}
\label{sec:mp_sim}


In this section, we present detailed simulation results with precise observations and plots for NetHEPT dataset, 
since we could do an extensive search for optimal budget split even for 5 phases, and also run a large number of Monte Carlo iterations for it.
%
%
As mentioned earlier, we also conduct simulations on Facebook dataset for upto 4 phases.
%
%
Unless specified, the results for these datasets qualitatively followed a very similar  pattern as that for the NetHEPT dataset.





\subsection{Distribution of the Extent of Diffusion}

All distributions corresponding to the extent of diffusion (for any number of phases or  for any amount of budget, at the end of any phase or within any intermediate phase) exhibit a bell-shaped nature.
\Cref{fig:dist} presents the
distributions of extents of diffusion over phases, for different number of phases with the corresponding optimal budget split, for NetHEPT (WC)  with $k=200$
(see \Cref{tab:budgetsplits} for optimal budget splits).
It can be notably seen that the means of the histograms are evenly spaced (e.g., for 3-phase diffusion, the mean extent after first phase equals the difference between the mean extents of second and first phase, which also equals the difference between the mean extents of third and second phase).
This has  implications as we will see in  \Cref{sec:result_budgetsplit}.


 \Cref{tab:Std200} presents the standard deviations of the overall extent of diffusion that happened till the end of each phase, and also standard deviations of extent of diffusion that happened during each phase (extent of diffusion at the end of the phase minus extent of diffusion at the beginning of the phase).
 It is important to note that we could compare the distributions across different value of phases and budget splits in a reliable way, because the set of underlying live graphs is common to all the simulations for a given network dataset.
 


\setlength\tabcolsep{.7mm}
\begin{table}[t]
\caption{Mean extents of diffusion using various budget splits on NetHEPT (WC) (optimal budget splits are highlighted)
}
\label{tab:budgetsplits}
\small
\begin{tabular}{|c|c||c|c||c|c|}
\hline
\T \B
 &  & \multicolumn{2} {c||} {$k=200$} & \multicolumn{2} {c|} {$k=50$}
 \\ 
\hline \T \B 	{Pha}	&	{Budget split}	&	{Budget split}	&	{Mean}	&	{Budget split}	&	{Mean}	\\
ses & ratio & & extent & & extent
\\
\hline
\hline \T \B 	$1$	&	$1$	&	$(200)$	&	$2389$	&	$(50)$	&	$947$	\\
\hline \T \B 	$2$	&	$1\!:\!1$	&	$(100,\!100)$	&	$2464$	&	$(25,\!25)$	&	$961$	\\
\hline \rowcolor{Gray} \T \B 	$2$	&	$1\!:\!2$	&	$(67,\!133)$	&	$2478$	&	$(17,\!33)$	&	$965$	 \\
\hline \T \B 	$3$	&	$1\!:\!1\!:\!1$	&	$(66,\!67,\!67)$	&	$2491$	&	$(16,\!17,\!17)$	&	$967$	\\
\hline \T \B 	$3$	&	$1\!:\!1\!:\!2$	&	$(50,\!50,\!100)$	&	$2499$	&	$(12,\!13,\!25)$	&	$969$	\\
\hline \T \B 	$3$	&	$1\!:\!2\!:\!1$	&	$(50,\!100,\!50)$	&	$2494$	&	$(12,\!25,\!13)$	&	$967$	\\
\hline \T \B 	$3$	&	$2\!:\!1\!:\!1$	&	$(100,\!50,\!50)$	&	$2481$	&	$(25,\!12,\!13)$	&	$963$	\\
\hline \T \B 	$3$	&	$1\!:\!2\!:\!6$	&	$(22,\!45,\!133)$	&	$2502$	&	$(6,\!12,\!32)$	&	$973$	\\
\hline \T \B 	$3$	&	$1\!:\!2\!:\!4$	&	$(28,\!58,\!114)$	&	$2506$	&	$(7,\!14,\!29)$	&	$973$	\\
\hline \T \B 	$3$	&	$3\!:\!2\!:\!4$	&	$(66,\!44,\!90)$	&	$2493$	&	$(16,\!11,\!23)$	&	$969$	\\
\hline \rowcolor{Gray} \T \B 	$3$	&	$1\!:\!2\!:\!3$	&	$(33,\!67,\!100)$	&	$2508$	&	$(8,\!17,\!25)$	&	$975$	\\
\hline \T \B 	$4$	&	$1\!:\!1\!:\!1\!:\!1$	&	$(50,\!50,\!50,\!50)$ 	&	$2509$	&	$(12,\!12,\!13,\!13)$	&	$974$	\\
\hline \T \B 	$4$	&	$1\!:\!2\!:\!6\!:\!18$	&	$(7,\!15,\!45,\!133)$	&	$2511$	&	$(2,\!4,\!11,\!33)$	&	$974$	\\
\hline \T \B 	$4$	&	$1\!:\!2\!:\!4\!:\!8$	&	$(14,\!27,\!54,\!105)$	&	$2515$	&	$(4,\!7,\!13,\!26)$	&	$978$	\\
\hline \rowcolor{Gray} \T \B 	$4$	&	$1\!:\!2\!:\!3\!:\!4$	&	$(20,\!40,\!60,\!80)$	&	$2519$	&	$(5,\!10,\!15,\!20)$	&	$982$	\\
\hline \T \B 	$5$	&	$1\!:\!1\!:\!1\!:\!1\!:\!1$	&	$(40,\!40,\!40,\!40,\!40)$	&	$2513$	&	$(10,\!10,\!10,\!10,\!10)$	&	$978$	\\
\hline \T \B 	$5$	&	$1\!:\!2\!:\!6\!:\!18\!:\!54$	&	$(3,\!6,\!16,\!45,\!130)$	&	$2514$	&	$(1,\!2,\!5,\!12,\!30)$	&	$980$	\\
\hline \T \B 	$5$	&	$1\!:\!2\!:\!4\!:\!8\!:\!16$	&	$(6,\!12,\!25,\!52,\!105)$	&	$2518$	&	$(2,\!4,\!7,\!13,\!24)$	&	$980$	\\
\hline \rowcolor{Gray} \T \B 	$5$	&	$1\!:\!2\!:\!3\!:\!4\!:\!5$	&	$(14,\!27,\!40,\!54,\!65)$	&	$2525$	&	$(4,\!7,\!10,\!13,\!16)$	&	$985$	\\
\hline 
\end{tabular}
\end{table}
\setlength\tabcolsep{2mm}

\begin{table}[t]
\caption{Standard deviations of extent of diffusion using optimal budget splits on NetHEPT (WC) with $k=200$}
\label{tab:Std200}
\centering
\small
\begin{tabular}{|c|c|c|}
\hline \T \B
Phases & {till end of each phase} & {during each phase} 
\\ \hline \T \B
$1$ & $(89)$ & $(89)$
\\ \hline \T \B
$2$ & $(87,96)$ & $(87,68)$
\\ \hline \T \B
$3$ & $(85,94,96)$ & $(85,72,55)$
\\ \hline \T \B
$4$ & $(75,94,97,97)$ & $(75,71,60,49)$
\\ \hline \T \B
$5$ & $(72,86,96,97,97)$ & $(72,69,56,50,42)$
\\ \hline
\end{tabular}
\end{table}

\begin{figure*}
\small
\begin{tabular}{p{.34\textwidth} p{.33\textwidth} p{.33\textwidth}}
\hspace{-7mm}
\includegraphics[scale=.47]{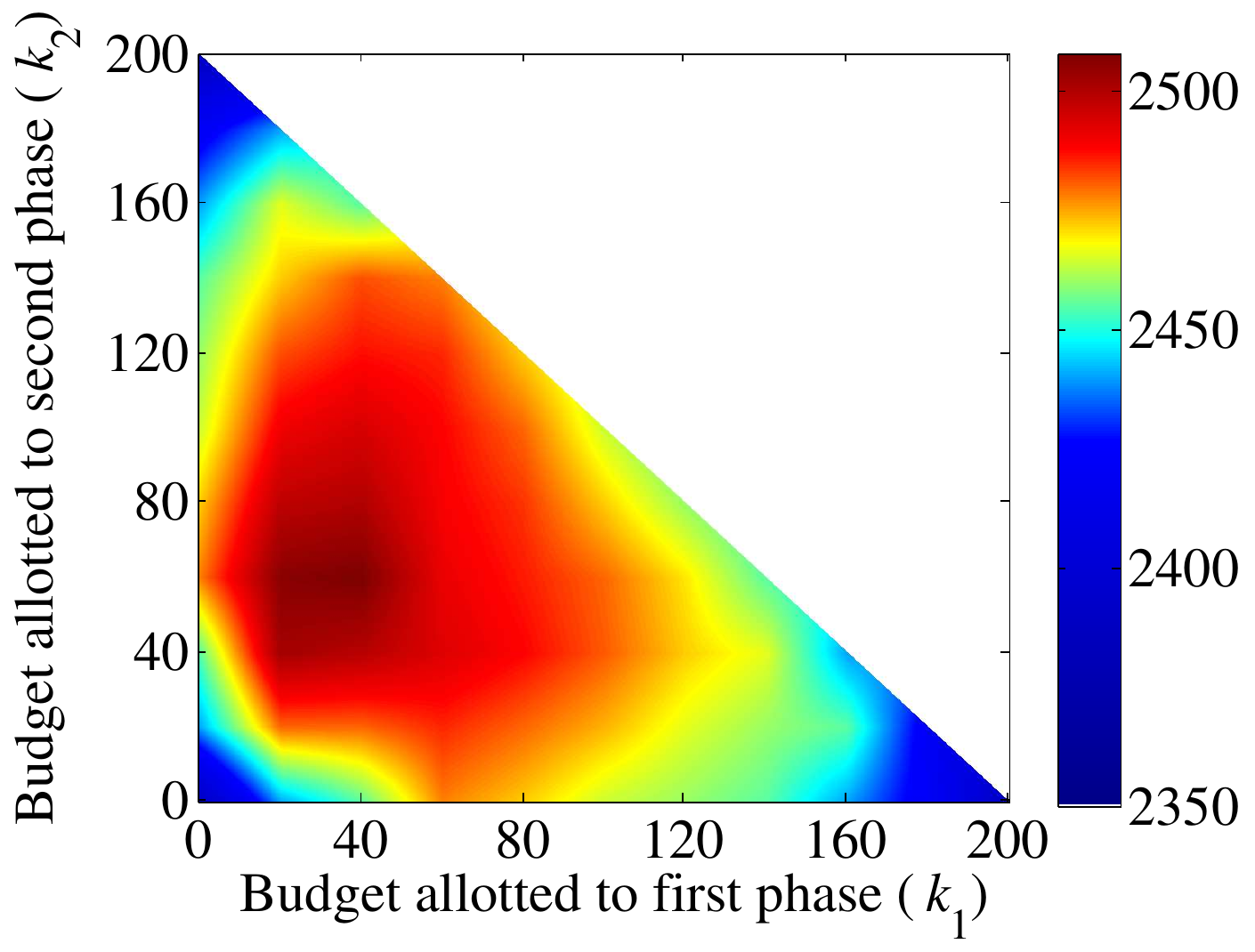}
&
\includegraphics[scale=.47]{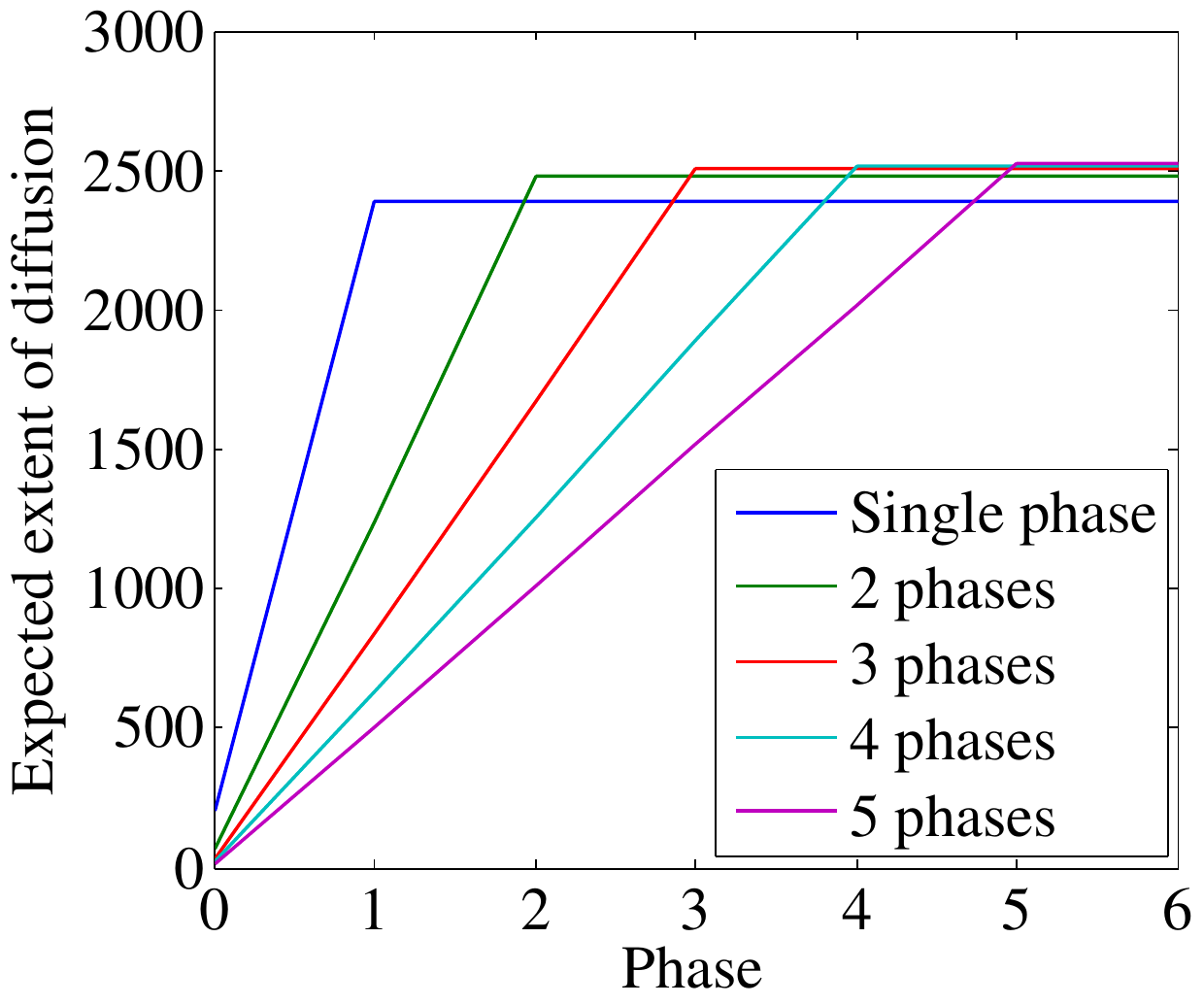}
&
\includegraphics[scale=.44]{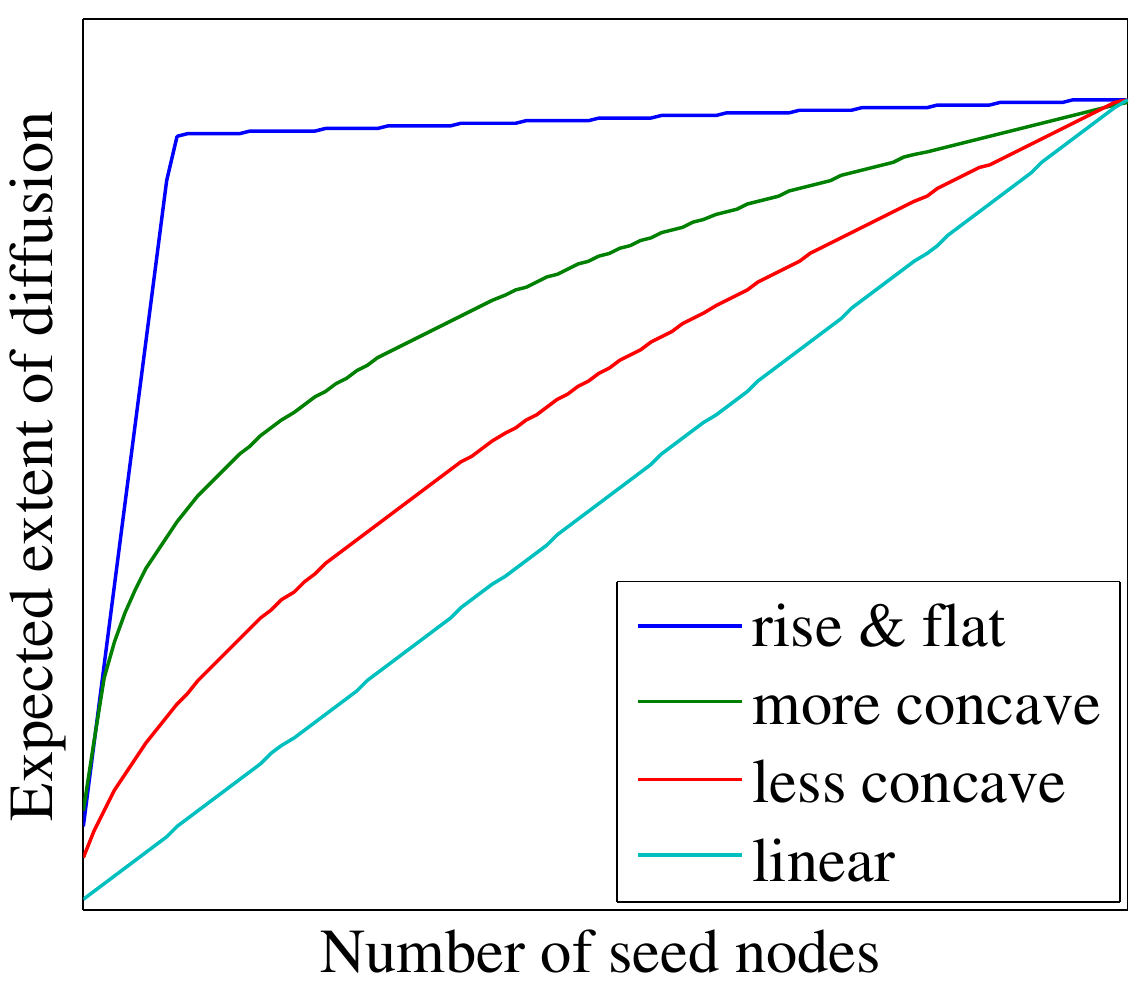}
\\
{(a) Effectiveness of various 3-phase budget splits for NetHEPT (WC) with $k=200$}
&
(b) Phasewise progression for different number of phases for NetHEPT (WC)  $k=200$
&
(c) Types of influenceability curves
\end{tabular}
\caption{Aspects of multiphase diffusion}
\label{fig:all}
\end{figure*}

We now try to numerically understand what it means by saying: using multiple phases would reduce uncertainty.
We first address the question: will multiple phases lead to a lower standard deviation at the end of the diffusion as compared to single phase with the same overall budget?
The answer is `no'.
%
It is true that for a low value in single phase for a bad live graph, the value improves for that live graph when we use multiple phases. But the value for a good live graph also improves. 
Moreover, multiphase diffusion would have a better reach than single phase diffusion, and would reach parts of the live graph which would stay unexplored in single phase for the same live graph. So the uncertainty of these newly explored parts also could get added to multiphase diffusion,
which may in fact may result in multiphase diffusion having a higher standard deviation than single phase.
This can be seen from \Cref{tab:Std200} where for instance, standard deviation at the end of the second phase of 2-phase diffusion ($96$) is greater than that at the end of single phase ($89$).
In general, these observations show that $p+1$ phases may actually lead to a higher standard deviation at the end of the diffusion as compared to $p$ phases with the same overall budget.



\Cref{tab:Std200} also shows that in a $p$-phase diffusion, the standard deviation of the extent of diffusion that progresses during phase $q+1$ is consistently less than that of the extent of diffusion that progresses during phase $q$.
For instance, in a $3$-phase diffusion, the standard deviation of the extent of diffusion that progresses during second phase ($72$) is less than that of the extent of diffusion that progresses during first phase ($85$).
One of the reasons for this observation could be a lower uncertainty in the extent of diffusion triggered by the selection of lesser influential seed nodes.



We now try to see what it  means when we say that `using multiple phases would reduce uncertainty as compared to single phase'.
To answer this more concretely, we quantify how the second phase reduces uncertainty in a multiphase diffusion, since this is the phase which first distinguishes multiphase diffusion from single phase.
For instance, to quantify how the second phase reduces uncertainty in a $p$-phase diffusion with budget split $(k_1,k_2,\ldots,k_p)$,
  could mean the following:
the standard deviation of the extent of diffusion in the second phase (say $\sigma_p$) would be less than the standard deviation of the additional number of nodes influenced using the single phase diffusion if the single phase diffusion budget  is increased from $k_1$ to $k_1+k_2$ (say $\sigma_s$).
%
We observe in our simulations that this statement holds true.
In particular, for NetHEPT (WC) with $k=200$, 
$\sigma_p$ for a given $p$ is the second component of the corresponding vector in the `during each phase' column of \Cref{tab:Std200}. These do not exceed $72$, while we observed that $\sigma_s$ was higher than $100$.

We can draw an insight from this discussion:
it is beneficial to select highly influential nodes in the initial phases since they would not only lead to a large extent of observed diffusion, but also high uncertainty, which can be then improved upon by selecting other nodes in subsequent phases. 





\subsection{Effectiveness with Number of Phases}

 
Tables \ref{tab:gain} and \ref{tab:facebook} present the gains achieved by  multiphase diffusion with various number of phases and  corresponding optimal budget splits for NetHEPT (WC) and Facebook (WC), respectively.
%
%
We also ran simulations for 10 phases with a number of manually chosen budget splits for NetHEPT and observed that the final extent of diffusion did not exceed an expected value of 2535, which is a 6.11\% gain over single phase.
%
%
So we conclude that there is significant gain when we move from single phase to 2 phases, and an appreciable  gain when we further move to 3 phases. But there is a sharp decline in the additional gain beyond 3 phases.


\begin{table}[t]
\caption{Gain achieved by using multiphase diffusion using optimal budget split for NetHEPT (WC) with $k=200$}
\label{tab:gain}
\centering
\small
\begin{tabular}{|c|c|c|c|c|}
\hline \T \B
Phases & $2$ & $3$ & $4$ & $5$
\\ \hline \T \B
\% gain over single phase & $3.73$ & $5.00$ & $5.44$ & $5.70$ 
\\ \hline \T \B
\% gain over one phase less & $3.73$ & $1.21$ & $0.44$ & $0.24$
\\ \hline
\end{tabular}
\end{table}

\setlength\tabcolsep{1.2mm}
\begin{table}
\caption{Optimal budget splits for Facebook (WC) with $k=50$}
\label{tab:facebook}
\centering
\small
\begin{tabular}{|c|c|c|c|}
\hline \T \B
Phases & $2$ & $3$ & $4$
\\ \hline \T \B
Optimal budget split & $(15,35)$ & $(5,15,30)$ & $(2, 7, 15, 26)$
\\ \hline \T \B
\% gain over single phase & $5.41$ & $6.26$ & $6.68$
\\ \hline \T \B
\% gain over one phase less & $5.41$ & $0.80$ & $0.40$
\\ \hline
\end{tabular}
\end{table}
\setlength\tabcolsep{2mm}

\subsection{Optimal Budget Split}
\label{sec:result_budgetsplit}


\Cref{fig:all}(a) presents a visualization of the expected extents of diffusion for various budget splits obtained using our coarse and fine searches for NetHEPT (WC) dataset (the actual discrete heatmap is smoothened for better visualization).
\Cref{tab:budgetsplits} presents the expected extents of diffusion for various 
manually chosen budget split ratios; the optimal budget splits for all phases are highlighted.
The budget splits for NetHEPT (TV) dataset with $k=200$ were also very similar.
For NetHEPT (TV) dataset with $k=50$, given the number of phases, a very large set of budget splits gave very similar and almost optimal results.
\Cref{tab:facebook} presents optimal budget splits for Facebook (WC) dataset.
For Facebook (TV), it was rather optimal to select very few nodes in the initial phases; in particular, it was optimal to choose one node each in the first three phases for the case of 4-phase diffusion.

We consistently observed that the optimal budget splits had a lower budget allotted to the initial phases than to the latter phases, that is, 
 $k_q < k_{q+1}$.
 As motivated earlier, the reason for finding an optimal budget split is to find an optimal balance between observation and exploitation. 
 As per our multiphase adaptation of IRIE, the most influential nodes get selected in the initial phases, while the lesser influential ones get selected in the later phases.
 Owing to the highly influential nature of the initially selected nodes, it suffices to select only a few of them in the initial phases to get a good enough observation of the diffusion, and then use the lesser influential ones to cover the parts of the network which could not be reached in the observed diffusion.
Thus, the budget allotted to the initial phases plays a  critical role to set a right balance between  observation and exploitation.



Our general observation is that an optimal budget split is the one which would lead to an almost equal number of nodes getting influenced in expectation, across the given number of phases.
This can been seen for NetHEPT (WC) dataset in \Cref{fig:all}(b), where the
expected extent of diffusion grows linearly with the number of intermediate phases elapsed.
%
From our studied datasets with different values of budgets, it was evidently the case that a good balance between the extent of observation of diffusion in earlier phases and exploitation by adaptively selecting seed nodes in later phases, was found when the expected extent of diffusion in all the phases was almost the same.
This was with the exception of Facebook (TV) for which the extent of diffusion was very high in earlier intermediate phases.
We  present a compelling insight behind these observations, for which we first introduce the notion of {\em influenceability curve\/} of a network.


%

\subsubsection{Influenceability Curve of a Network}

Given a diffusion model and a deterministic influence maximizing algorithm, a network would have a plot depicting the number of influenced nodes versus the number of seed nodes. 
This plot, in some sense, provides an indication of:
what fraction of best seed nodes budget would lead to what fraction of the maximum achievable extent of diffusion using that budget.
We call this plot as the 
{\em influenceability curve\/} of the network 
with respect to the given diffusion model and influence maximizing algorithm.
It can be seen in  the literature on maximizing information diffusion that,
this curve is concave for most real-world social networks.
The results on some popular datasets for single phase diffusion can be found in
\cite{wang2012scalable,jung2012irie}.

%

The influenceability curves for such real-world network datasets can be classified broadly into the following four types (illustration in \Cref{fig:all}(c)):
\begin{enumerate}
\item
Rise \& flat: Facebook~(TV), Epinions~(TV), Slashdot(TV)
\item
Very concave: Facebook~(WC), DBLP~(WC), Slashdot~(WC), Epinions~(WC), Arxiv~(TV) 
\item
Less concave: NetHEPT~(WC), NetHEPT~(TV), NetPHY~(WC), Arxiv~(WC)
\item
Linear: DBLP~(TV), Amazon~(WC), Amazon~(TV) 
\end{enumerate}


The influenceability curve is concave for most networks, which means that selecting first few nodes in earlier phases is enough to give a good enough observation of diffusion, which spares the possibility of selecting higher number of nodes in the later phases to exploit the observation.
For NetHEPT datasets (WC and TV) with a budget such as 50 or 200, it so happens that an equal distribution of extent of diffusion across phases would arise from a budget split which roughly follows an arithmetic progression; that is, the  split of 1:2 for 2 phases, 1:2:3 for 3 phases, etc., for a reasonable amount of budget. This explains the specific observation of \cite{dhamal2phase} concerning the 2-phase budget split for NetHEPT-like datasets.

Also, more concavity (\Cref{fig:all}(c)) means an even higher skew in the selection of nodes, since an even smaller number of nodes in initial phases could provide sufficient observation. In such cases, we would have optimal budget split such that  first phase budget allocation is considerably lower than that of subsequent phases (e.g., Facebook (WC) in \Cref{tab:facebook}).

If for some unconventional network, the influenceability curve is convex (possibly because no node individually is highly influential, but collectively larger seed sets become highly influential), this would mean that selecting a few nodes in earlier phases do not lead to a significant extent of diffusion and hence a poor observation, which would then lead to a not-so-good adaptive seeding in the later phases. So it would be better to select a higher number of nodes in the earlier phases, which collectively could provide a good enough observation.
As a middle ground between concave and convex, if  influenceability curve is close to linear, the budget could be split equally across the phases.

If the influenceability curve rises and flattens (e.g., Facebook (TV)), which means that very few nodes are extremely influential, it would be well advised to not select these nodes in one phase itself.
%
Using multiple phases, it is possible to ascertain whether a highly influential node, when selected as seed node, influences another highly influential node without having to select the latter as a seed node.
In the single phase case, the influence maximizing algorithm would have selected the latter, by computing that the latter is highly influential but perhaps not very likely to be influenced by other selected seed nodes.
However, if in our observed diffusion, the latter is indeed influenced without having to select it as a seed node, it could help save our budget which could be used for selecting other seed nodes.
If it is not influenced in our observed diffusion, we select it as seed node in a following phase. 
Combining these two cases, we would gain in expectation by not selecting the very highly influential nodes in   first phase itself, if we have more than 2 phases at hand.

We can develop a simple method for determining optimal budget split based on these insights.
First, note that the difference between the performances of single phase and multiphase diffusion is not extremely high.
That is, the mean extent of diffusion in first phase of a multiphase diffusion in which the budget allotted to first phase is $k_1$ would be close to
the mean extent of diffusion for single phase with budget $k_1$,
the mean extent of diffusion in second phase of a multiphase diffusion in which the budget allotted to second phase is $k_2$ would be close to
the mean extent of additional diffusion for single phase when budget increases from $k_1$ to $k_1+k_2$.
So finding a budget split which leads to the mean extents of diffusion across intermediate phases to be equal,  is almost equivalent to partitioning the influenceability curve into $p$ pieces so that the mean extent of diffusion is split equally across these pieces.
E.g., when the number of phases is $p$ with total budget $k$, we look at the curve plot for number of seed nodes ranging from $0$ to $k$, and split the plot into $p$ equally spaced $Y$-coordinates. We look at corresponding $X$-coordinates (inverse function of influenceability curve) to obtain  values of $(k_1,k_1+k_2,\ldots,k)$ and hence derive our optimal budget split $(k_1,k_2,\ldots,k_p)$.
Rose \& flat curves being an exception, where we select one node in each non-terminal phase and  remaining nodes in  terminal phase.

\subsubsection{Additional Notes}

The expected extents of diffusion over various budget splits follow a somewhat unimodal behavior, as can be seen from  \Cref{fig:all}(a). So a multidimensional golden section search technique, if adapted well, has a good chance of finding optimal budget splits quickly.
Alternatively, instead of using our search technique for finding optimal budget split, it may be advantageous to use an iterative algorithm that converges to an optimal solution. Such algorithms usually require a good starting point for finding the optimal solution and also converging to it quickly.
Our results suggest that a budget split of $(k_q)_{q=1}^p$, where $k_q \leq k_{q+1}$
could act as a good starting point (attributed to the concave influeneability curves of real-world social networks).

\subsection{A Note on Value Decaying over Phases}

\Cref{fig:all}(b) presents the phasewise progression of multiphase diffusion.
There have been studies considering the decaying value of product or information over time
 \cite{zhang2016influence,dhamal2phase}.
 In our study, since a phase starts after the conclusion of the previous phase (which generally takes considerable number of time steps in IC model), it would almost certainly be disadvantageous to use multiple phases with product value decaying in every time step. Hence our study considers that the value of the product decays over phases rather than in every time step.
 Consider that the value of a node influenced in phase $q$ has a decay factor of $\delta^q$, where $\delta \in [0,1]$. That is, a node influenced in a later phase  provides a lesser value. 
So  given a budget split $\mathbf{K}$, the value of diffusion can be defined as 
$\sum_{q=1}^p \delta^q \beta_q(\mathbf{K})$, where $\beta_q(\mathbf{K})$ is the number of nodes influenced in phase $q$.
%
%
It is clear that lower values of $\delta$ are deterrent to using multiple phases. 
%
Given the number of phases and the corresponding optimal budget split, we could determine the value of $\delta$ below which, it would be rather advantageous to use single phase.



As per our hypothesis, an optimal budget split corresponds to a budget split which leads to equal number of nodes influenced in each of the individual phases (except for networks with `rise \& flat' type of influenceability curve).
%
%
Given that the number of phases is $p$, let $x_p$ be the number of nodes influenced in each of the $p$ phases. Since the spread achieved using single phase is lower than that achieved with $p>1$, let $\epsilon_p$ be the fractional loss incurred by using single phase instead of $p$ phases.
For $p$ phases (with the derived optimal budget split) to be advantageous over single phase, even with the incorporation of $\delta$, we should have
%

\begin{small}
\begin{align*}
& p \, x_p -\epsilon_p \, x_p \leq x_p(1+\delta+\ldots+\delta^{p-1})
\\
\Longleftrightarrow \;& p-\epsilon_p \leq \frac{1-\delta^p}{1-\delta}
\\
\Longleftrightarrow \;& \delta^p - p \, \delta + (p-1-\epsilon_p) \leq 0
\end{align*}
\end{small}
\vspace{2mm}

With $p$ and $\epsilon_p$ known, the above inequality can be easily solved to determine the range of $\delta$ subject to $\delta \in [0,1]$.
There even exist explicit solution formulae of the above inequality for $p$ upto 4 (and it is perhaps an overkill to have more than 4 phases).
Since we know $p$ and $\epsilon_p$ for different values of $p$, solving the above polynomial  gives the minimum value of $\delta$ for which, using $p$ phases (with the derived optimal budget split) holds advantage over single phase.
\Cref{tab:mindelta} presents the minimum values of $\delta$ for multiphase diffusion with the optimal budget splits as given in \Cref{tab:budgetsplits} to be advantageous over single phase, for  NetHEPT (WC) with $k=200$.

Note that there may exist a better budget split which could be advantageous as compared to single phase for a value of $\delta$ lower than the one found using  above analysis.
However, the above analysis guarantees that there exists a budget split (the previously found optimal budget split) which would lead to a better value than single phase, when the value of $\delta$ is higher than the one found using the above analysis (except for networks with `rise \& flat' type of influenceability curve).




We also did  preliminary analyses for values of $\delta=0.1,\ldots,0.9$ for identifying optimal budget split (from among the budget splits that we explored in our previous simulations) with respect to the value of diffusion while accounting for the decay factor.
We did a search over different budget splits explored in our earlier simulations, for which we already had stored the phasewise extents of diffusion
(we computed the value of a budget split by taking a weighted sum of the extents of diffusion in each phase $q$ with the weighing factor $\delta^q$).
Our observations suggest that for $\delta = 0.8,0.9$,
a budget split 
such that the number of nodes influenced in phase $p$ is approximately proportional to $\delta^p$, is optimal or near-optimal. Since the value of each node influenced in phase $p$ is $\delta^p$, we have that the total value obtained in phase $p$ is approximately proportional to $\delta^{2p}$.
The results were not very consistent for lower values of $\delta$.
We defer a more elaborate study in this direction for future work.



\begin{table}[t]
\caption{Minimum values of $\delta$ for NetHEPT (WC) with $k=200$ for multiphase with optimal budget splits (\Cref{tab:budgetsplits}) to be beneficial}
\label{tab:mindelta}
\centering
\small
\begin{tabular}{|c|c|c|c|c|}
\hline \T \B
Phases & $2$ & $3$ & $4$ & $5$
\\ \hline \T \B
Minimum $\delta$ & $0.810$ & $0.871$ & $0.904$ & $0.924$
\\ \hline
\end{tabular}
  \vspace{3mm}
\end{table}

\section{Conclusion
}

The objective of our study was to quantify and understand the effectiveness of multiphase diffusion in social networks under IC model.
We started by present a negative result that more phases do not guarantee a better extent of diffusion.
We used computationally efficient techniques for reducing the number of diffusion states in our simulations as well as searching for an optimal budget split. 
We studied the effect on the means and standard deviations of the extent of diffusion in different phases, and provided insight behind the reduction in uncertainty when we use multiple phases.
We also suggested using highly influential nodes in the initial phases since they would not only lead to a large extent of observed diffusion, but also high uncertainty, which can be then improved upon by other nodes in subsequent phases.
Our experiments suggested a significant improvement in spread when we move from single to two phases, but the marginal gain beyond three phases was usually not very significant. 
With the primary reasoning behind multiphase diffusion being able to observe in earlier phases and exploit in later phases, we observed that for most types of networks, a good balance between the two is found when the expected extent of diffusion in all the phases is almost the same.
We then presented a method for arriving at an optimal budget split using the influenceability curve of the network.
We concluded by considering a setting with decaying value of diffusion over phases, and provided a bound on the decay factor as a function of the number of phases; the multiphase diffusion would be advantageous over single phase if the value of decay factor is higher than this bound.

\subsection{Future Work}

It would  be useful to design efficient algorithms specifically for multiple phases to study larger datasets.
An elaborate study on value of diffusion decaying over phases, is warranted.
A game theoretic study would be interesting, where competing companies implement seeding in multiple phases.
The work can be extended to other diffusion models.
A multi-phase study over evolving social networks is an important practical aspect to look at.
A more theoretical study would  help lay foundation for studying multiphase information diffusion.




\vspace{2mm}
\bibliographystyle{IEEEtran}
\bibliography{Multiphase_references}  
%
%

\end{document}